\documentclass[10pt, journal]{IEEEtran}
\usepackage{amsfonts}
\usepackage{slashbox}
\usepackage{multirow}
\usepackage{graphicx}
\usepackage{amsfonts}
\usepackage{amsmath,epsfig}
\usepackage{mathbbold}
\usepackage{stfloats}
\usepackage{url}
\usepackage{bm}
\usepackage{cite}
\usepackage{latexsym}
\usepackage{epsfig}
\usepackage{graphics}
\usepackage{mathrsfs}
\usepackage{algorithmic}
\usepackage{algorithm}
\usepackage{subfigure}
\usepackage[all]{xy}
\usepackage{pifont}
\usepackage{bbding}
\usepackage{stmaryrd}
\usepackage{amsfonts}
\usepackage{epic}
\usepackage{latexsym}
\usepackage{epstopdf}
\usepackage{epic}
\usepackage{xcolor}
\usepackage{mathrsfs}
\usepackage{pifont}
\usepackage{epsfig}
\usepackage{mathbbold}
\usepackage{subfigure}

\usepackage{enumerate}
\usepackage{setspace}
\usepackage{latexsym}
\usepackage{booktabs}
\usepackage{color}
\usepackage{cases}
\usepackage{amsthm}
\usepackage{amssymb}

\newtheorem{proposition}{Proposition}

\newcommand{\mv}[1]{\mbox{\boldmath{$ #1 $}}}

\newcommand{\A}{\bm A}

\newcommand{\HH}{\bm H}

\newcommand{\V}{\bm V}

\newcommand{\aaa}{\bm{a}}
\newcommand{\N}{\mathcal{N}}

\newcommand{\II}{\mathcal{I}}
\newcommand{\G}{\bm G}
\newcommand{\Ss}{\bm S}

\newcommand{\W}{\bm W}

\newcommand{\uuu}{\bm u}

\newcommand{\vvv}{\bm v}

\newcommand{\g}{\bm g}

\newcommand{\ttheta}{\mathbf \Theta}

\newcommand{\K}{\mathcal{K}}

\newcommand{\E}{\mathcal{E}}

\graphicspath{{./figs/}}

\begin{document}
\title{Weighted Sum Power Maximization for Intelligent Reflecting Surface Aided  SWIPT}
\author{\IEEEauthorblockN{Qingqing Wu,  \emph{Member, IEEE} and Rui Zhang, \emph{Fellow, IEEE}   }
\thanks{The authors are with the Department of Electrical and Computer Engineering, National University of Singapore, email:\{elewuqq, elezhang\}@nus.edu.sg.}   }

\maketitle
\vspace{-0.5cm}
\begin{abstract}
The low efficiency of far-field wireless power transfer (WPT) limits the fundamental rate-energy (R-E) performance trade-off of the  simultaneous wireless information and power transfer (SWIPT) system. To address this challenge,  we propose in this letter a new SWIPT system aided by the emerging  intelligent reflecting surface (IRS) technology. By leveraging massive low-cost passive elements that are able to reflect the signals with adjustable phase shifts, IRS achieve a high passive beamforming gain, which is appealing for drastically enhancing  the WPT efficiency and thereby the R-E trade-off of SWIPT systems.  We consider an IRS being deployed to assist a multi-antenna access point (AP) to serve multiple information decoding receivers (IDRs) and energy harvesting receivers (EHRs).   We aim to maximize the weighted sum-power received by EHRs via jointly optimizing the transmit precoders at the AP and reflect phase shifts at the IRS, subject to the individual signal-to-interference-plus-noise ratio (SINR) constraints for IDRs. Since this  problem is non-convex,   we propose efficient  algorithms to obtain suboptimal solutions for it. In particular, we prove  that it is sufficient to send information signals only at the AP to serve both IDRs and EHRs regardless of their channel realizations. Moreover,  simulation results  show significant performance gains achieved by our proposed designs over benchmark schemes.

\end{abstract}
\vspace{-0.05cm}

\vspace{-0.4cm}
\section{Introduction} 
Radio frequency (RF) transmission enabled simultaneous wireless information and power transfer (SWIPT) is   a promising solution for future wireless powered Internet-of-Things (IoT) \cite{wu2016overview}. However, in practice, a typical energy harvesting receiver (EHR) requires much higher receive power (say, tens of dB more) than that needed for an information decoding receiver (IDR).
 As such, how to improve the efficiency of wireless power transfer (WPT) for EHRs is a critical challenge to resolve for fundamentally enhancing the rate-energy (R-E) trade-off of SWIPT systems \cite{xu2014multiuser}.


Recently, intelligent reflecting surface (IRS) comprised of  a large number of low-cost passive reflecting elements has been proposed as a competitive technology  to improve the spectrum and energy efficiency of future wireless networks \cite{JR:wu2019IRSmaga,wu2018IRS,JR:wu2018IRS,JR:wu2019discreteIRS}. In particular,   the phase shift induced by each reflecting element at the IRS can be adjusted so as to collaboratively  achieve high  passive beamforming gains. A comprehensive overview of IRS-aided  wireless networks was provided in \cite{JR:wu2019IRSmaga}  by including the IRS signal model, practical hardware architecture, signal processing/channel estimation as well as network deployment. Furthermore, it was shown in  \cite{JR:wu2018IRS} that IRS is capable of creating not only a ``signal hot spot'' but also a virtually ``interference-free'' zone in its vicinity via joint active beamforming at the access point (AP) and passive beamforming at the IRS. In particular, an asymptotic squared signal-to-noise ratio (SNR) gain is derived in \cite{JR:wu2018IRS}, which is larger than the linear beamforming gain  in both massive MIMO and MIMO relay systems \cite{JR:wu2018IRS,zhang2016fundamental}. This is because the passive signal reflection of IRS is generally noise-free and full-duplex, thus combining  the functionalities of both receive and transmit arrays for receive and reflect beamforming.
While prior works on IRS or its equivalent technologies (e.g.  \cite{wu2018IRS,JR:wu2019IRSmaga,JR:wu2018IRS,JR:wu2019discreteIRS,dongfang2019,yanggang2019,han2018large,huangachievable,yan2019passive}) mainly focus on exploiting IRS for wireless information transmission,  the high  beamforming gain brought  by low-cost IRS is also appealing for WPT \cite{JR:wu2019IRSmaga}.  In particular, without the need of any transmit RF chains, IRS enables  intelligent signal reflection over a large aperture to compensate the high RF signal attenuation over long distance and thereby creates an effective energy harvesting/charging zone in its neighborhood,  as illustrated in Fig. 1. However, to our best knowledge, IRS-aided SWIPT systems have not been studied in the literature yet,  which motivates this work.

\begin{figure}[!t]
\centering
\includegraphics[width=0.45\textwidth]{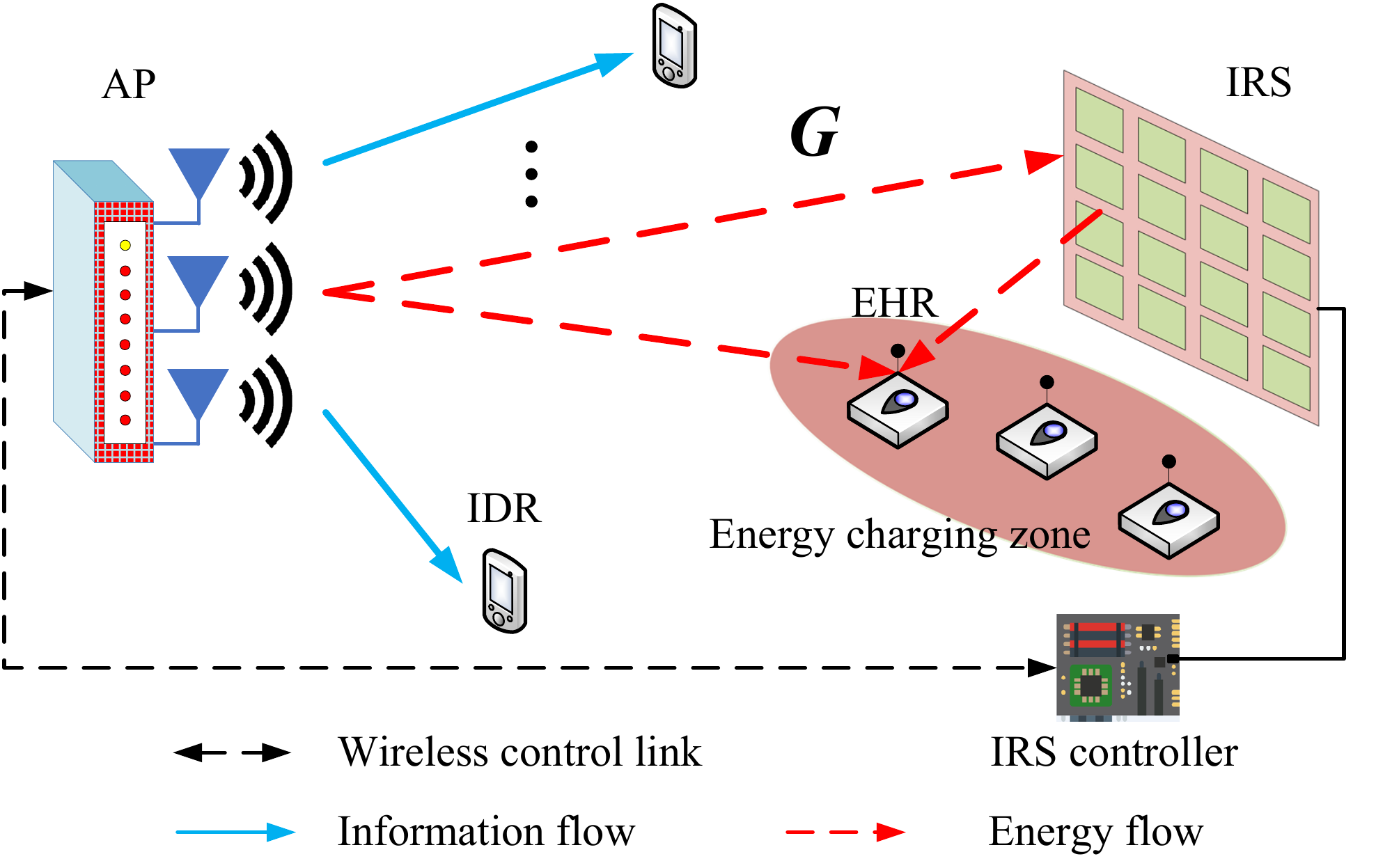}\vspace{-0.05cm} 
\caption{An IRS-assisted SWIPT system. } \label{system:model}\vspace{-0.2cm}
\end{figure}

As shown in Fig. \ref{system:model},  we study  an IRS-assisted SWIPT system with multiple antennas at the AP and single antenna at each of multiple  IDRs and EHRs. The IRS is deployed to  improve the WPT efficiency for a cluster of EHRs (e.g., IoT sensors or tags) located in an energy charging zone within the signal  coverage of the IRS. For low-complexity implementation, it is assumed that the IDRs do not possess the capability of cancelling the interference from energy signals transmitted by the AP (if any). Under the above setup, we  aim to maximize the weighted sum-power received by EHRs while satisfying a set of individual   signal-to-interference-plus-noise ratio (SINR) targets at IDRs, via jointly optimizing the active transmit beamforming vectors (precoders) at the AP and passive reflect phase shifts at the IRS. Since information signals for IDRs can be exploited for energy harvesting at EHRs, a fundamental question is whether dedicated energy signals are required for our considered system.  Note that for the conventional  SWIPT system  without IRS, it was shown in  \cite{xu2014multiuser} that dedicated energy signals are not needed, if the channels of all users (both EHRs and IDRs) are statistically independent. However, this assumption does not hold in general for IRS-aided SWIPT systems due to the additional signal path reflected by the IRS via arbitrary phase shifts. Thus, this result needs to be re-examined in our considered new SWIPT system with IRS.



Interestingly, we prove  that dedicated  energy signals are not required even for the IRS-aided SWIPT system, thus extending the above important result in \cite{xu2014multiuser} to the case with arbitrary user channels (provided that there is at least one IDR present in the system).
 Furthermore,  we propose efficient algorithms to obtain suboptimal solutions for the formulated problem, which is non-convex and thus difficult to solve optimally. Numerical results show significant performance gains achieved by our proposed designs and reveal useful insights on the practical deployment of IRS for SWIPT applications.

\emph{Notations:} $\mathbb{C}^{x\times y}$ denotes the space of $x\times y$ complex-valued matrices. For a vector $\bm{x}$, $\|\bm{x}\|$ denotes its Euclidean norm and  $[\bm{x}]_n$ denotes its $n$-th element.
The distribution of a circularly symmetric complex Gaussian (CSCG) random variable with mean ${x}$ and variance ${\sigma}^2$ is denoted by  $\mathcal{CN}({x},{\sigma}^2)$; and $\sim$ stands for ``distributed as''. For a square matrix $\Ss$, ${\rm{tr}}(\Ss)$  denotes its trace and $\Ss\succeq \bm{0}$ means that $\Ss$ is positive semi-definite.  For any general matrix $\A$, $\A^H$ and ${\rm{rank}}(\A)$ denote its conjugate transpose and rank, respectively. $\mathbb{E}(\cdot)$ denotes the statistical expectation. $ \mathrm{Re}\{\cdot\}$ denotes the real part of a complex number. For a set $\mathcal{N}$, $|\mathcal{N}|$ denotes its cardinality.

\section{System Model and Problem Formulation}
\subsection{System Model}
As shown in Fig. \ref{system:model}, we consider an IRS-assisted wireless network where an IRS with $N$ reflecting elements is deployed to assist in the SWIPT from the AP with $M$ antennas to two sets of single-antenna receivers, i.e.,  IDRs  and EHRs, denoted by the sets $\K_{\II}=\{1, \cdots,K_{I}\}$ and $\K_{\E}=\{1, \cdots,K_{E}\}$, respectively. The set of reflecting elements is denoted by $\mathcal{N}$ with $|\mathcal{N}|=N$.    For simplicity, we consider linear precoding at the AP and assume that each IDR/EHR is assigned with one dedicated information/energy beam without loss of generality. Thus, the transmitted signal from the AP is given by
\begin{align}
\mv{x} = \sum_{i\in {\mathcal{K_I}}}{\mv w}_i s_i^{\rm{ID}} +\sum_{j\in {\mathcal{K_E}}}{\mv v}_j s_j^{\rm{EH}}, \label{equa:jnl:1}
\end{align}
where ${\mv w}_i\in {\mathbb C}^{M\times 1}$ and ${\mv v}_j\in {\mathbb C}^{M\times 1}$ are the precoding vectors for IDR $i$ and EHR $j$, respectively, with the information-bearing and energy-carrying signals denoted by $s_i^{\rm{ID}}$ and $s_j^{\rm{EH}}$, satisfying $s_i^{\rm{ID}} \sim \mathcal{CN}(0,1), \forall i\in \mathcal{K_{I}}$ and   $\mathbb{E}\left(|s_j^{\rm{EH}}|^2\right)=1,\forall j\in \mathcal{K_{E}}$ \cite{xu2014multiuser}.
 Suppose that the AP has a total transmit power budget  $P$; from (\ref{equa:jnl:1}) we thus have $\mathbb{E}(\mv{x}^H\mv{x}) = \sum_{i\in {\mathcal{K_I}}}\|{\mv w}_i \|^2 +\sum_{j\in {\mathcal{K_E}}}\|{\mv v}_j\|^2 \le P$.

We assume  a quasi-static flat fading channel model and the channel state information (CSI) of all links are assumed to be perfectly known at the AP, for the purpose of characterizing the optimal R-E trade-off of the considered SWIPT system. Denote by $\bm{h}^H_{d,i}\in \mathbb{C}^{1\times M}$ and  $\bm{h}^H_{r,i}\in \mathbb{C}^{1\times N}$  the baseband equivalent channels from the AP to IDR $i$ and from
 the IRS to IDR $i$, respectively. Their counterpart channels for EHR $j$ are denoted by $\bm{g}^H_{d,j}$ and  $\bm{g}^H_{r,j}$, respectively, and the channel  from the AP to IRS is denoted by $\bm{G}\in \mathbb{C}^{N\times M}$.  Let  $\ttheta  = \text{diag} (\beta_1 e^{j\theta_1}, \cdots, \beta_N e^{j\theta_N})$ denote the (diagonal) reflection-coefficient matrix at the IRS, where   $\beta_n \in (0, 1]$ and $\theta_n\in [0, 2\pi)$ denote the reflection amplitude and  phase shift of the $n$th element, respectively \cite{JR:wu2019IRSmaga,JR:wu2018IRS}.  In this paper, we set $\beta_n=1$, $n\in \N$ to maximize the signal reflection of the IRS.  Thus, the  signal received at IDR $i$ is given by\footnote{Although the IRS is deployed mainly for enhancing the WPT for EHRs, we take into account the IRS-reflected signals at  IDRs without loss of generality.}  
\begin{align}
{y}_i^{\rm{ID}} = {\mv h}^H_i\mv{x} + z_i,\ \forall i\in \mathcal{K_{I}}, \label{equa:jnl:2}
\end{align}
where  ${\bm{h}}^H_i= \bm{h}^H_{r,i}\ttheta\bm{G}+\bm{h}^H_{d,i}$ and $z_i\sim \mathcal{CN}(0,\sigma_i^2)$ is the receiver noise. Under the assumption that IDRs do not possess the capability of cancelling the interference caused by energy signals,  the SINR of IDR $i$, $i\in \mathcal{K_{I}}$, is given by
\begin{align}\label{eq:SINR}
\text{SINR}_i = \frac{|{\bm{h}}^H_i\bm{w}_i |^2}{\sum\limits_{ k\neq i, k\in \K_{\II} }|{\bm{h}}^H_i\bm{w}_k |^2 + \sum\limits_{ j \in \K_{\E} }|{\bm{h}}^H_i\bm{v}_j |^2   + \sigma^2_i}.
\end{align}
On the other hand, by ignoring the noise power, the received RF power at EHR $j$, denoted by $E_j$, is given by
\begin{align}\label{EH:energy}
E_j= \sum\limits_{k\in\mathcal{K_I}}|\g^H_j {\mv w}_k|^2  + \sum \limits_{k\in\mathcal{K_E}}|\g^H_j {\mv v}_k|^2,\ \forall j\in \mathcal{K_{E}},
\end{align}
where $\g^H_j = \bm{g}^H_{r,j}\ttheta \bm{G} +  \bm{g}^H_{d,j}$.
\vspace{-0.2cm}
\subsection{Problem Formulation}
We aim to maximize the weighted sum-power received by EHRs subject to the individual SINR constraints at different IDRs, given by $\gamma_i, i\in \K_{\II}$.  Denote by  $\alpha_j\geq 0$ the energy weight of EHR $j$ where a larger value of $\alpha_j$ indicates a higher priority for sending  energy to EHR $j$ as compared to other EHRs. Let $\Ss =\sum_{j\in  \K_{\E}}\alpha_j\g_j\g^H_j$. Based on  \eqref{EH:energy}, the
weighted sum-power received by all EHRs can be expressed as
\begin{align}
\sum \limits_{j\in\mathcal{K_E}} \alpha_j E_j =  \sum\limits_{i\in\mathcal{K_I}}{\mv w}_i^H{\mv S}{\mv w}_i  + \sum\limits_{j\in\mathcal{K_E}}{\mv v}_j^H{\mv S}{\mv v}_j.
\end{align}
Accordingly,  the optimization problem is formulated as
\begin{align}
\!\!\!\!\text{(P1)}: \max_{\{\bm{w}_i\}, \{\bm{v}_j\},\bm{\theta}} & \sum_{i\in \K_{\II}} \bm{w}_i^H\Ss\bm{w}_i + \sum_{j\in \K_{\E}} \bm{v}_j^H\Ss\bm{v}_j   \label{eq:obj}\\
\mathrm{s.t.}~~~~&{\text{SINR}}_{i}\geq \gamma_i, \forall i \in \K_{\II}, \label{P1:SINRconstrn}\\
&\sum_{i\in \K_{\II}}\|\bm{w}_i\|^2 + \sum_{j\in \K_{\E}}\|\bm{v}_j\|^2\leq P,  \\
& 0\leq \theta_n \leq 2\pi, \forall n\in\mathcal{N}. \label{phase:constraints}
\end{align}
Since  the transmit precoders and IRS phase shifts are  intricately coupled in both the objective function and SINR constraints,  (P1) is a non-convex optimization problem that is challenging to solve.  Furthermore, it remains unknown whether sending dedicated energy beams (i.e., $\bm{v}_j$'s, $\forall j\in \mathcal{\K_\E}$) for  EHRs is necessary to attain the optimality of (P1) when there is at least one IDR in presence, i.e., $\bm{v}_j={\bm 0}, \forall j\in \mathcal{\K_\E}$ or not.

\vspace{-0.2cm}
\section{Special Case: IRS-Aided  WPT}
To draw useful insight, we first   consider a special case of (P1) where there exist only the EHRs (i.e., without any IDRs). In this case, (P1) is simplified as
\begin{align}
\text{(P2)}: ~~\max_{\{\bm{v}_j\}, \bm{\theta}} ~~~&\sum_{j\in \K_{\E}} \bm{v}_j^H\Ss\bm{v}_j    \label{eq:obj}\\
\mathrm{s.t.}~~~~& \sum_{j\in \K_{\E}}\|\bm{v}_j\|^2\leq P,  \\
& 0\leq \theta_n \leq 2\pi, \forall n\in\mathcal{N}. \label{phase:constraints}
\end{align}
However, (P2) is still non-convex due to the non-concave objective function with respect to $\bm{v}_j$'s and $\bm{\theta}$. Nevertheless, we observe that by fixing either the energy precoders or phase shifts, the resultant  problem of (P2)  can be solved efficiently, which thus motivates us to apply alternating optimization to solve (P2) by iteratively optimizing $\bm{v}_j$'s and  $\bm{\theta}$ until the convergence is reached, with the details given  as follows.

First, for any fixed $\bm{\theta}$, it can be shown  that the optimal energy precoders of all EHRs should align with  the principle eigenvector of $\Ss$ that corresponds to its largest eigenvalue \cite{xu2014multiuser}, denoted by ${\bar \vvv}_0(\Ss)$. Thus, without loss of optimality, only one common energy beam needs to be sent  for all EHRs and this energy precoder is given by $\vvv^*_0=\sqrt{P}\bar \vvv_0(\Ss)$.
In particular, if $\alpha_k\gg\alpha_j, \forall j\neq k$, we have $\Ss \approx \alpha_k\g_k\g^H_k$ and   $\vvv^*_0=\sqrt{P}\g_k/\|\g_k\|$, which means that the AP should beam the energy signal towards EHR $k$ exactly to maximize its receive  RF power.
%
Second, for any fixed $\vvv^*_0$, (P2) is reduced to
\begin{align}\label{optimize:theta}
\max_{ \bm{\theta}} ~~~&\sum_{j\in\mathcal{K_E}}\alpha_j |( \bm{g}^H_{r,j}\ttheta \bm{G} +  \bm{g}^H_{d,j})\bm{v}^*_0|^2\\  
\mathrm{s.t.}~~~~& 0\leq \theta_n \leq 2\pi, \forall n\in\mathcal{N}. \label{phase:constraints2}
\end{align}
Let  $\bm{g}^H_{d,j}\bm{v}^*_0 ={b}_{j}$ and  $\bm{g}^H_{r,j}\ttheta\bm{G}\bm{v}^*_0 =\bm{u}^H\bm{a}_{j}$ where $\bm{u} = [u_1, \cdots, u_N]^H$,  $u_n =e^{j\theta_n}, \forall n$, and $\bm{a}_{j}=\text{diag}(\bm{g}^H_{r,j})\bm{G}\bm{v}^*_0$, $\forall j$. Constraints in \eqref{phase:constraints2} are equivalent to $|u_n|=1, \forall n\in \mathcal{N}$. Thus, problem \eqref{optimize:theta} is equivalent  to
\begin{align}\label{optimize:theta2}
\max_{ \bm{u}} ~~~&\sum_{j\in\mathcal{K_E}} \alpha_j | {\bm u}^H {\bm a}_j+ b_j|^2 \\
\mathrm{s.t.}~~~~&|u_n|=1, \forall n\in\mathcal{N}. \label{unit:modulus:constraint}
\end{align}
Although  the constraints in the above problem are non-convex, we note that the objective function is convex with respect to $\uuu$, which allows us to apply the successive convex approximation (SCA) technique for solving  it.  Let $\A=\sum_{j\in \K_{\E}} \alpha_j \aaa_{j}\aaa_{j}^H$.  Since any convex function is lower-bounded globally  by its first-order Taylor expansion at any feasible point,  for a given local point $\hat{\uuu}$, the quadratic term $\uuu^H\A\uuu$ in \eqref{optimize:theta2} satisfies
\begin{align}
\uuu^H\A\uuu \geq 2\mathrm{Re}\{ \uuu^H  \A \hat{\uuu}\} - \hat{\uuu}^H\A\hat{\uuu},
\end{align}
 where  the equality holds at the point $\hat{\uuu}$. As such, the objective function of problem \eqref{optimize:theta2} is lower-bounded  by
\begin{align}\label{SCA:obj}
\!\! 2\mathrm{Re}\{ \uuu^H ( \A \hat{\uuu} \!+\!\!\! \sum_{j\in \K_{\E} }\alpha_j \aaa_jb^H_j ) \}\!-\!  \hat{\uuu}^H\A\hat{\uuu} \!+ \!\!\! \sum_{j\in \K_{\E} }\alpha_j |b_j|^2.
\end{align}
Based on \eqref{SCA:obj}, it is not difficult to show that the optimal solution is given by ${u}_n^*= 1$ if $\eta_n=0$ and ${u}_n^*= \frac{\eta_n}{|\eta_n|}$ otherwise, where  $\eta_n= {[\A \hat{\uuu} +\sum_{ j\in \K_{\E} }\alpha_j \aaa_jb^H_j]_n}, \forall n$. 

Since the objective value of (P2) is non-decreasing by alternately optimizing the energy precoder and phase shifts,  and it is also upper-bounded by a finite value, the proposed algorithm is guaranteed to converge. Furthermore,  this algorithm is of low complexity since the optimal solution in each iteration admits  a closed-form expression.


\section{IRS-Aided SWIPT: Energy Beamforming or Not?}
Next, we study the general case where there is at least one IDR coexisting with EHRs. Since the information signal can be utilized at EHRs  for energy harvesting, it remains unknown whether dedicated energy signals are required for solving (P1) optimally. We address  this important question in this section and propose an efficient algorithm to solve (P1) by extending the alternating optimization algorithm in Section III.


{Let ${\mv W}_i={\mv w}_i{\mv w}_i^H ,\forall i \in \mathcal{K_I}$
and ${\mv W}_\E=\sum_{j \in \mathcal{K_E}}\mv v_j\mv v_j^H$. Then, it follows that ${\rm{rank}}({\mv W}_i)\le 1,\forall i \in \mathcal{K_I}$
and ${\rm{rank}}({\mv W}_\E)\le {\min}(M,K_E)$. By ignoring the above rank constraints on ${\mv W}_i$'s and ${\mv W}_\E$ similarly as in  \cite{xu2014multiuser}, the semidefinite relaxation (SDR) of (P1), denoted by (SDR1), is given by}
\begin{align}
\!\!\max_{\{\mv{W}_i\},\mv{W}_\E, \bm{\theta}}~
& \sum\limits_{i\in\mathcal{K_I}}{\rm{tr}}(\Ss\mv{W}_i)+{\rm{tr}} (\Ss\mv{W}_\E) \\
\text{s.t.} ~~~~~~&  \frac{ {\rm{tr}}(\mv{h}_i\mv{h}_i^H\mv{W}_i)}{\gamma_i}-\sum\limits_{k\neq i,k\in\mathcal{K_I}}{\rm{tr}}(\mv{h}_i\mv{h}_i^H\mv{W}_k) \nonumber \\
&~~~ - {\rm{tr}}(\mv{h}_i\mv{h}_i^H\mv{W}_\E)-\sigma_i^2 \geq 0,  \forall i \in \mathcal{K_I},\\ &
\sum\limits_{i\in\mathcal{K_I}}{\rm{tr}}(\mv{W}_i)+ {\rm{tr}}(\mv{W}_\E)\leq P,\\
~& {{\mv W}_i}\succeq {\mv 0}, \forall i\in \mathcal{K_I}, ~~{{\mv W}_\E}\succeq {\mv 0}.
\end{align}
Note that for  any fixed $\bm{\theta}$,  (SDR1) is reduced to the same problem as that  in \cite{xu2014multiuser}. However, the result in \cite{xu2014multiuser} that there are no energy beams needed for optimally solving the problem  is not applicable to (SDR1) here. This is because  due to the signal reflection by the IRS with arbitrary phase shifts $\bm{\theta}$, the effective  channels of different users cannot be assumed to be statistically independent in general as in  \cite{xu2014multiuser}.
In addition, even if both  EHRs and IDRs are very far from the IRS such that  $ {\bm{h}}^H_i\approx \bm{h}^H_{d,i}, i\in \mathcal{K_I}$ and $ {\bm{g}}^H_j\approx \bm{g}^H_{d,j}, j \in \mathcal{K_E}$, these channels may still be correlated due to e.g., line-of-sight (LoS) propagation.
This thus motivates the following proposition.  Let the optimal solution of  (SDR1) be $\W^*_i, i\in \K_{\II}$ and $\W^*_{\E}$ for any fixed  $\bm{\theta}$ (i.e., fixed ${\bm{h}}^H_i, i\in \mathcal{K_I}$ and ${\bm{g}}^H_j, j \in \mathcal{K_E}$).
\begin{proposition}\label{rank:result}
For arbitrary user channels ${\bm{h}}^H_i, i\in \mathcal{K_I}$ and ${\bm{g}}^H_j, j \in \mathcal{K_E}$,  problem (SDR1) always has an optimal solution  satisfying  $\W^*_{\E} ={\bm 0}$ and ${\rm{rank}}(\W^*_i)=1$, $\forall i$.
\end{proposition}
\begin{proof}
Please refer to Appendix A.
\end{proof}

Proposition \ref{rank:result} extends the result in \cite{xu2014multiuser} by showing that even if the users' channels are not statistically  independent, sending  information signals only is sufficient for achieving the maximum weighted sum-power at EHRs. This is expected since sending dedicated energy signals not only consumes transmit power but also potentially causes interference to IDRs.  Based on Proposition \ref{rank:result},  for any given $\bm{\theta}$,   $\W^*_i$'s  can be obtained by solving (SDR1) with $\W^*_{\E}={\bm 0}$ via  standard solvers, e.g., CVX. Then the transmit precoders ${\bm{w}}^*_i$'s can be recovered by performing  eigenvalue decomposition over  the obtained rank-one $\W^*_i$ 's.


%

Next, for given  transmit precoders $\bm{w}^*_i$'s,   (P1) is reduced to
\begin{align}\label{sectIV:theta}
\max_{ \bm{\theta}} ~~~&\sum_{j\in \K_{\E}}\alpha_j   \sum_{i\in K_{\II}} |( \bm{g}^H_{r,j}\ttheta \bm{G} +  \bm{g}^H_{d,j})\bm{w}^*_i|^2\\
\mathrm{s.t.}~~~~&{\text{SINR}}_{i}\geq \gamma_i, \forall i \in \K_{\II}, \\
& 0\leq \theta_n \leq 2\pi, \forall n\in\mathcal{N}.  \label{sectIV:unit}
\end{align}
Similarly as in Section III, let  $\bm{g}^H_{d,j}\bm{w}^*_i ={b}_{j,i}$,    $\bm{h}^H_{d,k}\bm{w}^*_i ={d}_{k,i}$,   $\bm{g}^H_{r,j}\ttheta\bm{G}\bm{w}^*_i  =\bm{u}^H\bm{a}_{j,i}$, and  $\bm{h}^H_{r,k}\ttheta\bm{G}\bm{w}^*_i =\bm{u}^H\bm{c}_{k,i}$    where  $\bm{a}_{j,i}=\text{diag}(\bm{g}^H_{r,j})\bm{G}\bm{w}^*_i$ and $\bm{c}_{k,i}=\text{diag}(\bm{h}^H_{r,k})\bm{G}\bm{w}^*_i$, $k,i   \in \K_{\II}, j\in \K_{\E}$.   As a result, problem \eqref{sectIV:theta} is equivalent  to
{\begin{align}\label{sectIV:theta2}
\max_{ \bm{{u}} }~&\sum_{j\in \K_{\E}}\alpha_j   \sum_{i\in K_{\II}} | \bm{u}^H\bm{a}_{j,i}   +          {b}_{j,i}          |^2 \\
\mathrm{s.t.}~~& \frac{|\bm{u}^H\bm{c}_{i,i} + d_{i,i}  |^2}{\sum_{k\neq i, k\in K_{\II} }|\bm{u}^H\bm{c}_{i,k} + d_{i,k}  |^2 +  \sigma^2_i}\geq \gamma_i,  \forall i, \\
& |u_n|=1, \forall n\in\mathcal{N}.
\end{align}}
Since ${\bm u}$ is involved in both the objection function and constraints, the SCA technique in Section III is not applicable for solving  problem \eqref{sectIV:theta2}.
Let $\bm{\bar{u}}\!=\! [\bm{{u}}; t]$, $\bm{R}^{\E}_{j,i} =[ \bm{a}_{j,i}\bm{a}_{j,i}^H, \bm{a}_{j,i}b^H_{j,i};    \bm{a}_{j,i}^Hb_{j,i},    0]$, and  $\bm{R}^{\II}_{i,k} =[ \bm{c}_{i,k}\bm{c}_{i,k}^H,   \bm{c}_{i,k}d^H_{i,k};    \bm{c}_{i,k}^Hd_{i,k}, 0]$ where $t$ is an  auxiliary variable.
Note that $\bm{\bar{u}}^H\bm{R}^{\E}_{j,i} \bm{\bar{u}}={\rm{tr}}(\bm{R}^{\E}_{j,i} \bm{\bar{u}}\bm{\bar{u}}^H)  $ and  $\bm{\bar{u}}^H\bm{R}^{\II}_{i,k} \bm{\bar{u}}={\rm{tr}}(\bm{R}^{\II}_{i,k} \bm{\bar{u}}\bm{\bar{u}}^H)$.
Define $\bm{V}=\bm{\bar{u}}\bm{\bar{u}}^H$, which needs to satisfy  $\bm{V}\succeq {0}$ and ${\rm{rank}}(\bm{V})=1$. Since the rank-one constraint is non-convex, we drop  this constraint and relax problem \eqref{sectIV:theta2} as
\begin{align}
\max_{\V} ~~~& \sum_{j\in \K_{\E}}\alpha_j   \sum_{i\in K_{\II}}  \left( {\rm{tr}}( {\bm{R}^{\E}_{j,i}}{\bm V})  + |b_{j,i}|^2 \right)  \label{SDP:V} \\
\mathrm{s.t.}~~~~&{\rm{tr}}( {\bm{R}^{\II}_{i,i}}{\bm V})  + |d_{i,i}|^2 \geq \gamma_i\sum_{k\neq i, k\in K_{\II} }   {\rm{tr}}( {\bm{R}^{\II}_{i,k}}{\bm V})  \nonumber \\
~~~~&~~~~~~~~+\gamma_i\sum_{  k\neq i, k\in K_{\II}  } |d_{i,k}|^2 + \gamma_i\sigma^2_i,   \forall i\in K_{\II},\label{P6:SINR:39}\\
~~~~& \bm{V}_{n,n} = 1, \forall n=1,\cdots, N+1,~~~ \bm{V} \succeq 0.  \label{P6:C9}
\end{align} 
Note that problem \eqref{SDP:V} is a  convex semidefinite program (SDP) that can be optimally solved by CVX.
However, it may not lead to a rank-one solution in general, i.e., ${\rm{rank}}(\V)\neq 1$, which implies that the optimal objective value
of problem \eqref{SDP:V} only serves an upper bound of problem \eqref{sectIV:theta}. In this case, Gaussian randomization can be used  to recover a high-quality feasible solution of  problem \eqref{sectIV:theta} similarly as in \cite{JR:wu2018IRS}. The algorithm stops when the fractional increase of the objective function of (P1) by iteratively optimizing information precoders and phase shifts is below a sufficiently small threshold or problem \eqref{SDP:V} becomes infeasible.


\vspace{-0.3cm}
\section{Numerical Results}
This section provides numerical results to validate our proposed design. The signal attenuation at a reference distance of 1 meter (m) is set as 30 dB  for all channels. Since the IRS is usually   deployed to avoid severe signal blockage with the AP and its assisted energy charging zone is practically of small size, the pathloss exponents of both the AP-IRS and IRS-user channels are set to be 2.2, which is lower than that of the AP-user channels assumed equal to 3.6, due to the random locations of the users (both EHRs and IDRs). To account for small-scale fading, we assume Rayleigh fading for the IRS-user and AP-user channels.  As the IRS reflects signals only in its front half-sphere, each reflecting element is assumed to have a 3 dBi gain.  Other required parameters are set as follows unless specified otherwise: $\sigma_i^2=-90$\,dBm,   $\gamma_i=\gamma, \forall i\in \K_{\II}$, $\alpha_j=1, \forall j\in \K_{\E}$, $M=4$, and $N=50$.

 \begin{figure}[!t]
\centering
\includegraphics[width=0.5\textwidth]{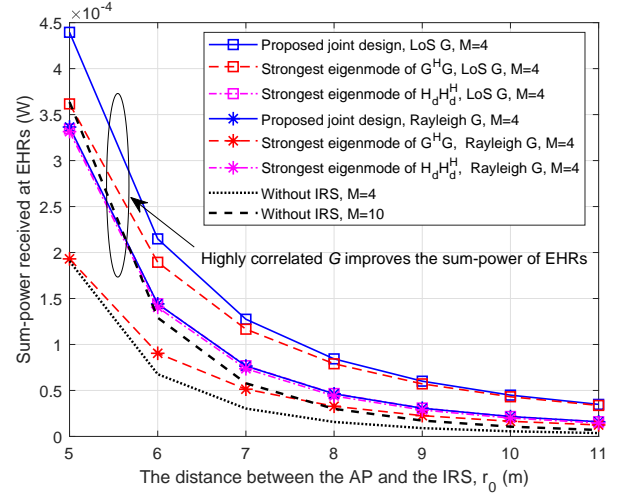}\vspace{-0.1cm}
\caption{Sum-power of EHRs versus the AP-IRS distance. } \label{simulation:distance}\vspace{-0.2cm}
\end{figure}

\subsection{IRS-Aided WPT}
We assume that all EHRs lie between the IRS and the AP with a distance of 2 m from the IRS. The distance between the AP and the IRS is denoted by $r_0$ m.  For comparison, we consider  three benchmark schemes: 1) Strongest eigenmode of $\G^H\G$ where $\vvv^*_0=\sqrt{P}{\bar \vvv}_0(\G^H\G)$, 2) Strongest eigenmode of $\HH_d\HH^H_d$ where  ${\HH}_d=[{\bm h}_{d,1}, \cdots,{\bm h}_{d,K_E}]$ and  ${ \vvv}^*_0=\sqrt{P} {\bar \vvv}_0(\HH_d\HH^H_d)$, and 3) Without IRS where ${ \vvv}^*_0=\sqrt{P} {\bar \vvv}_0(\HH_d\HH^H_d)$. For 1) and 2), the phase shifts are optimized by the proposed design in Section III.  To draw useful insight on the IRS deployment, we consider two cases of $\G$, i.e., $\G$ with all elements being 1 or  $\G$ with all elements following $ \mathcal{CN}(0,1)$ independently, which correspond to  deploying the IRS in an LoS-dominated and a rich-scattering environment (with Rayleigh fading) in practice, respectively.

 In Fig. \ref{simulation:distance}, we plot the received  sum-power of EHRs versus $r_0$ with $K_E=4$.  It is observed that by deploying the IRS around EHRs, the proposed design with $M=4$ can significantly improve the  sum-power of EHRs in both cases of $\G$ as compared to the case without using IRS (for both $M=4$ and $M=10$), and also outperforms other benchmark schemes. Furthermore, one can observe that the LoS channel between the AP and the IRS enables the EHRs to receive more RF energy than that under the Rayleigh fading channel between them. This fundamentally differs from the result in IRS-aided information transmission that the LoS AP-IRS channel generally degrades the performance due to the rank deficiency of $\G$ and the resultant more severe multiuser interference \cite{JR:wu2018IRS}. In contrast, as only one energy beam is sent from the AP for WPT,  channel correlation among  EHRs due to rank-one $\G$ allows them to simultaneously harvest more energy, thus making the energy beamforming more effective.

 \begin{figure}[!t] 
\centering
\includegraphics[width=0.5\textwidth]{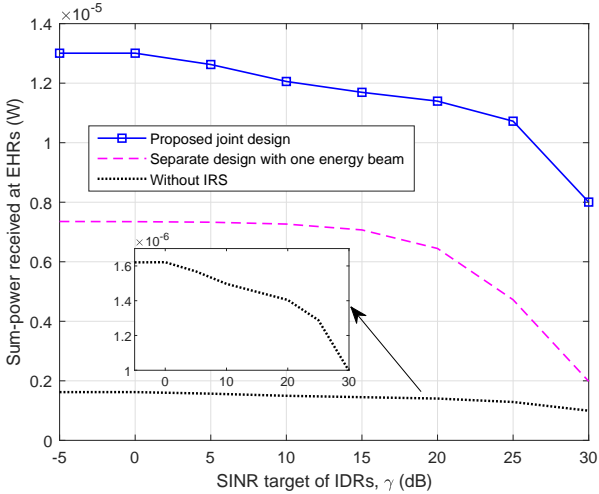}\vspace{-0.1cm}
\caption{Sum-power of EHRs versus SINR target of IDRs. } \label{simulation:tradeoff}\vspace{-0.2cm}
\end{figure}

\subsection{IRS-Aided SWIPT}
For the general case with both IDRs and EHRs, we plot  in Fig. \ref{simulation:tradeoff} the trade-off curve between the received  sum-power of EHRs and the achievable common SINR of IDRs with $K_I=2$ and $K_E=2$.
Motivated by Fig. \ref{simulation:distance}, we consider that the AP-IRS channel is dominated by LoS links with a distance of 15 m.   In addition, the EHRs lie between the IRS and the AP with a distance of 3 m from the IRS,
whereas the IDRs are 50 m away from the AP.  Besides the case without using  IRS, the scheme in \cite{xu2014multiuser} with separately designed information/energy beams is also adopted for comparison. For this scheme, the information beams are designed to minimize the total transmit power required for satisfying all the SINR constraints of IDRs, while the remaining AP transmit power is used to send one  energy beam, denoted by ${\bm v}_0$, to maximize the received  sum-power of EHRs, subject to the constraint without causing any  interference to all IDRs (provided that $K_I \le M-1$), i.e., adding the constraint $\bm{h}^H_{d,i}{\mv v}_0 ={0},\forall  i \in \mathcal{K_I}$ in (P2).

From Fig. \ref{simulation:tradeoff}, it is observed that the achievable power-SINR region of the SWIPT system can be significantly enlarged by deploying the IRS. For example, without compromising the SINR of  IDRs, the  received RF power at EHRs is greatly improved. 
Furthermore, one can observe that the benchmark  design with one dedicated   energy beam suffers considerable performance loss as compared to the proposed design. This is expected since according to Proposition 1, dedicated energy signals should not be used in our considered system.

\section{Conclusion}
This paper studied the weighted sum-power maximization problem in an IRS-assisted SWIPT system. We proposed efficient algorithms to obtain suboptimal solutions for the joint AP precoding and IRS phase shift optimization. In particular,  it was shown that dedicated  energy signal is not required for the general SWIPT system with arbitrary user channels.   Numerical results verified that the proposed design with IRS is able to drastically enlarge the R-E performance trade-off of SWIPT systems. It was also shown that deploying the IRS in strong LoS with  the AP is beneficial for improving the harvested energy of EHRs and thereby the system R-E trade-off.
\vspace{-0.1cm}
\appendices
\section*{Appendix A:  Proof of Proposition 1}  \label{apdx:B}
To prove Proposition \ref{rank:result}, we first introduce the problem formulation without ${{\mv W}_\E}$, denoted by (SDR2), i.e.,
\begin{align}
\mathop{{\max}}_{\{\mv{W}_i\}, \bm{\theta}}~
& \sum\limits_{i\in\mathcal{K_I}} {\rm{tr}}(\Ss\mv{W}_i) \\
\text{s.t.}~~~~&  \frac{{\rm{tr}}(\mv{h}_i\mv{h}_i^H\mv{W}_i)}{\gamma_i}-\sum\limits_{k\neq i,k\in\mathcal{K_I}}{\rm{tr}}(\mv{h}_i\mv{h}_i^H\mv{W}_k),\\
&~~~~~~~~~~ -\sigma_i^2 \geq 0,  \forall i \in \mathcal{K_I},\\ &
\sum\limits_{i\in\mathcal{K_I}}{\rm tr}(\mv{W}_i)\leq P,\\
~& {{\mv W}_i}\succeq {0}, \forall i\in \mathcal{K_I}.
\end{align}
Note that for any given $\bm{\theta}$, ${\bm{h}}^H_i, i\in \mathcal{K_I}$ and ${\bm{g}}^H_j, j \in \mathcal{K_E}$ are fixed. Denote the corresponding optimal objective values of (SDR1) and (SDR2)  by $E^*_1$ and $E^*_2$, respectively.
Then we establish   the equivalence of (SDR1) and (SDR2) by proving  $E^*_1= E^*_2$, via showing 1) $E^*_1\geq E^*_2$ and 2) $E^*_1\leq E^*_2$, respectively.  The proof of $E^*_1\geq E^*_2$ is intuitive since (SDR2) is a special case of (SDR1) with $\mv{W}_\E={\bm 0}$. Thus, we focus on proving $E^*_1\leq E^*_2$ as follows. The key of the proof lies in showing that for any optimal solution to (SDR1), we can always construct a feasible solution to   (SDR2)  that achieves the same weighted sum-power as $E^*_1$.

Suppose that $\{ \mv{W}^*_i\text{'s},  \forall i\in \mathcal{K_I},  \mv{W}^*_{\E} \}$  is the optimal solution to (SDR1) corresponding to  $E_1^*$. By exploiting the Lagrange  duality, it was shown  in \cite{xu2014multiuser} that  the optimal energy covariance matrix should satisfy  
\begin{align}\label{eq:1}
\lambda_i {\rm{tr}}(\HH_i\mv{W}^*_\E)=0,  \forall i\in \mathcal{K_I},
\end{align}
where $\lambda_i$'s are the dual variables associated with the SINR constraints in (SDR1). Next,  we discuss the following two cases: 1) $\lambda_i=0, \forall i\in \mathcal{K_I}$, and 2) there exists at least one $m \in \mathcal{K_I}$ with $\lambda_m>0$. For case 1), it was proved in \cite{xu2014multiuser} that all information beams should align with ${\bar \vvv}_0(\Ss)$ given  in Section III  and thus no dedicated energy signal needs to be sent. For case 2), based on  \eqref{eq:1},  we have
\begin{align}\label{orthogonal:condition}
{\rm{tr}} (\HH_m\mv{W}^*_\E)=0.
\end{align}
The key insight from  \eqref{orthogonal:condition} is that we can always send dummy information beams for user $m$ to replace energy beams without affecting its SINR constraint.
To this end,  we construct another solution to (SDR2) $\{ \widetilde{\mv{W}}_i\text{'s},  \forall i\in \mathcal{K_I} \}$ with $\mv{\widetilde{W}}_i=\mv{W}^*_i, \forall i\neq m$, $\widetilde{\mv W}_m ={\mv W}^*_m+ {\mv W^*_\E}$. It is easy to verify that the newly constructed solution is  feasible for (SDR2) and  achieves the same objective value as that of  (SDR1), i.e.,
\begin{align}\label{}
\widetilde{E}_2 &= \sum_{i\neq m, i\in\mathcal{K_I}} {\rm{tr}}( \Ss \widetilde{\mv W}_i ) + {\rm{tr}}( \Ss \widetilde{\mv W}_m ) \nonumber \\
& = \sum_{i\in \mathcal{K_I} } {\rm{tr}}( \Ss{\mv W}^*_i ) + {\rm{tr}}( \Ss {\mv W}^*_\E )  = E_1^*.
\end{align}
 Since $\widetilde{E}_2 \leq {E}^*_2$, we have ${E}^*_1\leq {E}^*_2$.

 Based on the facts that  ${E}^*_2\leq {E}^*_1$ and ${E}^*_1\leq {E}^*_2$, we obtain ${E}^*_2= {E}^*_1$, which  suggests that we only need to solve  (SDR1) with $\W^*_{\E} ={\bm 0}$.
Furthermore, it was shown in \cite{xu2014multiuser} [Problem (SSDP) in Remark 3.1] that there always exists an optimal solution with all $\mv{W}_i$'s of rank one for (SDR1), which thus completes the proof.


%

\end{document}